\documentclass[submission,copyright,creativecommons]{eptcs}
\providecommand{\event}{PrePost16}

\usepackage{amsfonts}
\usepackage{amsmath}
\usepackage{amssymb}
\usepackage{amsthm}
\usepackage{tikz}
\usepackage{mathpartir}
\usepackage{subcaption}
\usepackage{xspace}
\usepackage{mathtools}
\usepackage{stmaryrd} 
\usepackage[strings]{underscore} 

\newtheorem{theorem}{Theorem}[section]

\newtheorem{example}[theorem]{Example} 
\newtheorem{definition}[theorem]{Definition}

\usepackage{pic}
\usepackage{types}
\usepackage{monitor}
\usepackage{shorthand}
\usepackage{mathpartir}
\usepackage[inference]{semantic}
\usepackage{color}  

\newcommand{\Srv}{\textsc{Srv}\xspace}
\newcommand{\Cli}{\textsc{Cli}\xspace}
\newcommand{\sat}[2]{\ensuremath{\textsf{sat}(#1,#2)}\xspace}
\newcommand{\srvPre}{\ensuremath{\sqsubseteq_\Srv}\xspace}

\title{Preliminary Results Towards Contract Monitorability}

\author{
Annalizz Vella
\institute{CS, ICT,  University of Malta}
\email{annalizz.vella.10@um.edu.mt}
\and
Adrian Francalanza
\institute{CS, ICT,  University of Malta}
\email{adrian.francalanza@um.edu.mt}
}

\def\titlerunning{Preliminary Results Towards Contract Monitorability}
\def\authorrunning{A.~Vella and A.~Francalanza}

\begin{document}

\maketitle

\begin{abstract}
  This paper discusses preliminary investigations on the monitorability of contracts for web service descriptions. There are settings where servers do not guarantee statically whether they satisfy some specified contract, which forces the  client (\ie the entity interacting with the server) to perform dynamic checks.  This scenario may be viewed as an instance of Runtime Verification, where a pertinent question is whether  contracts can be monitored for adequately at runtime, otherwise stated as the \emph{monitorability of contracts}.  We consider a simple language of finitary contracts describing both clients and servers, and develop a formal framework that describes server contract monitoring.  We define monitor properties that potentially contribute towards a comprehensive notion of contract monitorability and show that our simple contract language satisfies these properties.     
\end{abstract}

\section{Introduction}
\label{sec:introduction}

Web services \cite{Castagna:2009:TCW:1538917.1538920:Contracts,CarCasLanPad:2006:FACWS:Contracts} typically consist of two types of computing entities. \emph{Servers} offer ranges of sequences of \emph{service interactions} to \emph{clients}, which in turn interact with these services and occasionally reach a state denoting client satisfaction.  The service interactions offered by a server typically follow some predefined structure that may be formalised as a contract \cite{CarCasLanPad:2006:FACWS:Contracts,Castagna:2009:TCW:1538917.1538920:Contracts,Padovani:2009:CDAWS:contracts,DBLP:journals/corr/BernardiH15}.  Dually, the service interactions invoked by a client may also be expressed within the same formalism.  

The contract calculus defined in \cite{Padovani:2009:CDAWS:contracts,Bernardi:2012:MST:2245276.2232097:Contracts,DBLP:journals/mscs/BravettiZ09a} is an abstract formalism equipped with an operational semantics that provides an implementation-agnostic, high-level description of client-server interactions; this permits formal reasoning about web services, such as  whether a client is compatible with a server or whether a server is able to satisfy the service interactions requested by the client. Such reasoning may, for instance, be used by clients for \emph{dynamic service discovery}, where a client decides to interact with a server whenever the contract it advertises satisfies the requirements of the client.

\begin{example}
Consider the contract below  describing the behaviour of an internet banking server: 
\begin{align*}
\prf{\nam{login}}{\bigl(\isel{(\prf{\coact{\nam{valid}}}{(\esel{\prf{\nam{query}}{\nil}}{\prf{\nam{transfer}}{\nil}})}) }{(\prf{\coact{\nam{invalid}}}{\nil})}\bigr)}
\end{align*}
It states that the server first expects a \nam{login} service interaction followed by either a \nam{valid} or \nam{invalid} service invocation; the operator $\isel{}{}$ denotes that the server decides autonomously whether to invoke \nam{valid} or \nam{invalid} in response.  If it branches to the latter, it terminates all interactions, denoted by $\nil$.  However, if it internally decides to invoke the service interaction \nam{valid}, it then offers a choice (denoted by the symbol $\esel{}{}$) of service interactions: it either accepts (account balance) \nam{query} interactions or else  (fund) \nam{transfer} interactions.    A contract describing the behaviour of a possible bank client is given below:
\begin{align*}
  \prf{\coact{\nam{login}}}{\bigl( 
      \esel
          {(\prf{\nam{invalid}}{\prf{\coact{\nam{reason}}}{\ok}})}
          {\esel
            {(\prf{\nam{expired}}{\ok})}
            {(\prf{\nam{valid}}{\prf{\coact{\nam{query}}}{\ok}})}
          }
  \bigr)}
\end{align*}
After a login service invocation, this client expects either of three responses: an \nam{invalid} interaction prompting another service request that asks for a \nam{reason} why the login was invalid, a login \nam{expired} invocation  or else a \nam{valid} login interaction that is followed  by invoking a \nam{query} service request.    All these alternative sequences leave the client in a satisfied state, \ok.  By analysing the \resp contracts, one can deduce that interactions on the \nam{valid} service following a client \nam{login} interaction necessarily lead to a \nam{query} interaction, which then leaves the client satisfied. One can also discern that \nam{invalid} interactions lead to a deadlock, whereby the client asks for a \nam{reason} service that is not offered by the server.  One can also note that the \nam{expired} option offered by the client is never chosen by the server.   \exqed 
\end{example}

Within this framework, there still remains the question of whether a service behaviour actually adheres to the contract it advertises. In general, static techniques (such as session-based type systems \cite{DCDL:2009:SST:1880906.1880907:SessionTypes}, or state-based model-checking of compliance, must or fair testing inclusion  \cite{Padovani:2009:CDAWS:contracts,Bernardi:2012:MST:2245276.2232097:Contracts,DBLP:journals/mscs/BravettiZ09a}) are used to verify \emph{before deployment} whether a server implementation respects the contract that describes it.  However, there are cases where this solution is not applicable. For instance, the client may decide \emph{not} to trust the static verifier used by the server. Alternatively, in a dynamic setting where service components are downloaded and installed at runtime, pre-deployment checks cannot be made on the server implementation since some components only become available for inspection at runtime.  There are also cases whereby a server does not come equipped with a formal description at all.

In these circumstances, a client can check that a server respects an advertised (or expected) contract by analysing the behaviour exhibited by the server \emph{at runtime}.   There are a number of cases where such a solution is adopted \cite{BocchiCDHY:2013:MNMST:SessionTypesContracts,Jia:2016:MBA:2837614.2837662:MonitorsSessionTypes}, making use of dynamic monitoring, possibly in conjunction with  other verification techniques.  This monitoring of systems may be seen as an instance of  Runtime Verification (RV) \cite{Leucker:2009:rv}, a lightweight formal verification technique used to check the current execution of a program by verifying it against some properties. 
In a typical setup, the monitor observing the running system raises a flag when a \emph{conclusive} verdict is reached, denoting that the property being checked for is either \emph{satisfied} or \emph{violated}.

An important question in any RV setup is that of the \emph{monitorability} of the specification language considered.  Indeed, it is generally the case that not all aspects of a specification can be monitored for and determined at runtime, as shown in  \cite{CF:2015:LTLPSRV:RV-LTL,BLC:2011:RVL:2000799.2000800:RV-LTL,FrancalanzaAI:2015:uHML} for specification languages such as LTL and the modal $\mu$-calculus.  In this work, we start to investigate the monitorability of contracts which, in turn, sheds light on the viability and expressiveness of the dynamic contract checking setup discussed above. In contrast to earlier work on monitorability, we do \emph{not} rely on an external formal logic for specifying the properties expected by a server contract, \eg a satisfaction relation $p \models \phi$ where $\phi$ would be a formula from a logic defined over server contract $p$ through the semantic relation $\models$.  Instead, we use the subcontract server relation $q \srvPre p$ defined in \cite{Padovani:2009:CDAWS:contracts} as a \emph{refinement semantic relation} where $q$ is an abstract description of the expected properties of a server contract $p$, thus using the contract language itself as a specification language.  Within this setting, we investigate  whether our monitoring mechanism is expressive enough to verify whether a server $p$ indeed refines an abstract description $q$.   

The rest of the paper is structured as follows.  Section~\ref{sec:language} overviews our contract language and defines our notion of contract satisfaction.  Section~\ref{sec:monitorability} introduces our monitoring setup and Section~\ref{sec:preliminary-results} relates verdicts reached by our monitored computations to the contract satisfactions discussed in Section~\ref{sec:language}.  Section~\ref{sec:conclusion} concludes by discussing related and future work.

\section{Servers, Clients and Satisfaction}
\label{sec:language}
\begin{figure}[t]
    \textbf{Syntax}
     \begin{align*}
       p,q\in\Srv & \bnfdef \nil && (\text{inaction}) &
       & \bnfsepp \prf{\act}{p} && (\text{prefixing}) \\
       & \bnfsepp \esel{p}{q} && (\text{external choice}) &
      & \bnfsepp \isel{p}{q} && (\text{internal choice}) 
     \end{align*}
     \begin{align*}
       r,s\in\Cli & \bnfdef \nil 
       &&  \bnfsepp \prf{\act}{r} 
       && \bnfsepp \esel{r}{s}  
       && \bnfsepp \isel{r}{s} 
       && \bnfsepp \ok && (\text{success})
     \end{align*}
   \textbf{Dynamics}
   \begin{mathpar}
 \inference[\rtit{Act}]{}{\prf{\act}{p}  \traSS{\act} p} \and
  \inference[\rtit{SelL}]{p \traSS{\mu} p'}{\esel{p}{q}  \traSS{\mu} p'}
  \and  
  \inference[\rtit{SelR}]{q \traSS{\mu} q'}{\esel{p}{q}  \traSS{\mu} q'}
\\
  \inference[\rtit{ChoL}]{}{\isel{p}{q}  \traSS{\acttau} p}
  \and
  \inference[\rtit{ChoR}]{}{\isel{p}{q}  \traSS{\acttau} q}  
\end{mathpar}
\textbf{Interaction}
\begin{mathpar}
  \inference[\rtit{AsyS}]{p \traSS{\acttau} q}{r\paral p \traSS{\acttau} r \paral q} \and
  \inference[\rtit{AsyC}]{r \traSS{\acttau} s}{r\paral p \traSS{\acttau} s \paral p} \and
  \inference[\rtit{Syn}]{r \traSS{\coact{\act}} s & p \traSS{\act} q}{r \paral p \traSS{\acttau} s \paral q} 
\end{mathpar}
  \caption{Server and Client Syntax and Semantics}
  \label{fig:syntax-semantics}
\end{figure}

Figure~\ref{fig:syntax-semantics} describes the syntax and semantics of (finite) servers and clients. Let $a,b,c,d \ldots \in\Names$ be a set of names denoting interaction addresses. Let $\coact{\cdot}$ be a complementation operation on these names where we refer to the complement of $a$ as $\coact{a}$; the operation is an involution, where $\coact{\coact{a}} = a$. The set of actions $\act\in\Act = \bigl(\Names \cup \sset{\coact{a} \ |\ a \in \Names}\bigr)$ includes all names and their complement. Let $\tau$ be a distinct action \emph{not} in $\Act$ 
denoting \emph{internal} unobservable activity, where we let $\acttt \in \Act \cup \sset{\acttau}$. 
Servers, $p,q\in\Srv$, consist of either the terminated server \nil, a prefixed server $\prf{\act}{p}$ that first engages in interaction $\act$ and then behaves as $p$, an external choice $\esel{p}{q}$ that can either behave as $p$ or $q$ depending on the interactions it engages in, or an internal choice \isel{p}{q} that autonomously decides to either behave as $p$ or $q$.  Clients, $r,s\in\Cli$, have a similar structure but may also consist of the term \ok denoting contract fulfilment.   
The semantics of both servers and clients are given in terms of a Labelled Transition System (LTS) where the labelled transition relation $p \traSS{\acttt} q$ is defined as the least relation satisfying the rules in  Figure~\ref{fig:syntax-semantics}; the definition of the transition relation for clients $r \traSS{\acttt} s$ is analogous and thus elided.  The definition is standard and follows that of related languages such as CCS \cite{Milner:1989:book}. For instance, the term $\prf{\act}{p}$ transitions with (action) label \act to the continuation $p$;  if $p$ can engage in an interaction on $\mu$ and transition to $p'$, then an external choice term involving $p$, \eg $\esel{p}{q}$ may also transition to $p'$ after exhibiting action $\mu$; by contrast, an internal choice involving $p$, \eg $\isel{p}{q}$ may transition to $p$ without exhibiting an external action ($\acttau$ is used). 
Servers and clients may be composed together to form a system, $r \paral p$, so as to engage in a sequence of interactions. Interactions are also defined as an LTS over systems, through the rules \rtit{AsyS}, \rtit{AsyC} and \rtit{Syn} in Figure~\ref{fig:syntax-semantics}. As is standard, silent transitions by either server or client allow them to transition autonomously in a system. However, a client transition on an external action must be matched by a server transition on the (dual) co-action for the transition to occur in the \resp system, denoting client-server interaction.    
\emph{Computations} are sequences of system transitions $r_0 \paral p_0 \traSS{\acttau} \ldots  \traSS{\acttau} r_n \paral p_n$, denoted as $r_0 \paral p_0  \wreduc  r_n \paral p_n$; the sequence may be potentially empty, $n=0$, where \emph{no} transitions are made, in which case we have  $r_0=r_n$ and $p_0=p_n$.  A computation $r_0 \paral p_0  \wreduc  r_n \paral p_n$ is \emph{maximal} whenever $\not \exists r',p' \,\cdot\,  r_n \paral p_n \traSS{\acttau} r' \paral p' $.  
\begin{definition} \label{sec:client-satisfaction} A maximal computation, $r \paral p  \wreduc  s \paral q$, is \emph{successful}, whenever the client's contract is fulfilled, meaning that $s=\ok$.  A service $p$ \emph{satisfies} a client $r$, denoted as \sat{p}{r}, when \emph{every} maximal computation rooted at $r\paral p$ is successful. \exqed
\end{definition}
\begin{example} 
  The server 
  \begin{math}
    p = \esel{\prf{\coact{a}}{\nil}}{(\isel{\prf{b}{\prf{a}{\nil}}}{\prf{c}{\nil}})}
  \end{math} may either transition as 
  \begin{math}
    p \traSS{\coact{a}} \nil
  \end{math} using rules \rtit{Act} and \rtit{SelL} from Figure~\ref{fig:syntax-semantics}, or silently transition as 
   \begin{math}
    p \traSS{\acttau} \prf{b}{\prf{a}{\nil}} 
  \end{math}
   or 
   \begin{math}
    p \traSS{\acttau} \prf{c}{\nil} 
  \end{math}
  via rules \rtit{ChoL}, \rtit{ChoR} and  \rtit{SelR} from Figure~\ref{fig:syntax-semantics}.  
  It satisfies the client 
   \begin{math}
     r = \esel{\prf{\coact{b}}{\ok}}{\prf{\coact{c}}{\ok}}
   \end{math},
   denoted as \sat{p}{r},  because the only maximal computations possible are the following
   \begin{align*}
     & r \paral p \traSS{\acttau} r \paral  \prf{b}{\prf{a}{\nil}}  \traSS{\acttau} \ok \paral \prf{a}{\nil} &
     & r \paral p \traSS{\acttau} r \paral \prf{c}{\nil} \traSS{\acttau} \ok \paral \nil
   \end{align*}
   both of which are successful.  By contrast, server $p$ does \emph{not} satisfy client 
   \begin{math}
     \prf{\coact{b}}{\ok}
   \end{math}, 
   denoted as $\neg \sat{p}{\prf{\coact{b}}{\ok}}$, nor does it satisfy the clients 
  \begin{math}
     \esel{\esel{\prf{\coact{b}}{\ok}}{\prf{\coact{b}}{\nil}}}{\prf{\coact{c}}{\ok}}
   \end{math}
   and 
  \begin{math}
     \esel{\prf{\coact{b}}{\prf{c}{\ok}}}{\prf{\coact{c}}{\ok}}.
   \end{math}
   In each case, we can show this through the \emph{unsuccessful} maximal computations below.
   \begin{align*}
       & \prf{\coact{b}}{\ok} \paral p \traSS{\acttau} \prf{\coact{b}}{\ok} \paral \prf{c}{\nil}
     \qquad\qquad        \esel{\esel{\prf{\coact{b}}{\ok}}{\prf{\coact{b}}{\nil}}}{\prf{\coact{c}}{\ok}} \paral p \traSS{\acttau} \esel{\esel{\prf{\coact{b}}{\ok}}{\prf{\coact{b}}{\nil}}}{\prf{\coact{c}}{\ok}} \paral \prf{b}{\nil} \traSS{\acttau} \nil \paral \nil
     \\
     \qquad& \esel{\prf{\coact{b}}{\prf{c}{\ok}}}{\prf{\coact{c}}{\ok}} \paral p  \traSS{\acttau} \esel{\prf{\coact{b}}{\prf{c}{\ok}}}{\prf{\coact{c}}{\ok}} \paral  \prf{b}{\nil}  \traSS{\acttau} \prf{c}{\ok} \paral \nil  &\qquad\qquad\;\;\;\;\exqed
   \end{align*} 
\end{example}

\noindent
The satisfaction predicate \sat{-}{-} induces a natural preorder amongst servers.

\begin{definition}[Server Preorder \cite{Padovani:2009:CDAWS:contracts}]\label{def:server-pre}
  A server $p$ is a subcontract of server $q$, denoted as $p \srvPre q$,  whenever, for \emph{all} clients $r$, \sat{p}{r} implies \sat{q}{r}.  Dually, $q$ is referred to as a supercontract of $p$. \exqed 
\end{definition}

Intuitively, $p \srvPre q$ of Definition~\ref{def:server-pre} means that we can substitute a server $p$ by a server $q$, safe in the knowledge that any client satisfied by $p$ would not be affected. 

\begin{example}\label{ex:serv-preorder}  Definition~\ref{def:server-pre} allows us to establish a number of useful server (in)equations such as
  \begin{align*}
    & \isel{\prf{\coact{a}}{\nil}}{\prf{b}{\nil}} \srvPre \prf{\coact{a}}{\nil}  && \esel{\prf{b}{\prf{a}{\nil}}}{\prf{b}{\prf{c}{\nil}}}  \srvPre \prf{b}{(\isel{\prf{a}{\nil}}{\prf{c}{\nil}})}   && \prf{b}{(\isel{\prf{a}{\nil}}{\prf{c}{\nil}})}   \srvPre  \esel{\prf{b}{\prf{a}{\nil}}}{\prf{b}{\prf{c}{\nil}}}
  \end{align*}
  but also justify subtle cases where substituting one server for another might break client satisfaction.  For instance, we have $\nil \not\srvPre \prf{a}{\nil}$ because for the client $\esel{(\isel{\ok}{\ok})}{\prf{\coact{a}}{\nil}}$  we have \sat{\nil}{\esel{(\isel{\ok}{\ok})}{\prf{\coact{a}}{\nil}}} since
  ${\esel{(\isel{\ok}{\ok})}{\prf{\coact{a}}{\nil}} \paral \nil \traSS{\acttau}  \ok \paral \nil}$ is the only maximal  computation (which is also successful), \emph{but} also have $\neg\sat{\prf{a}{\nil}}{\esel{(\isel{\ok}{\ok})}{\prf{\coact{a}}{\nil}}}$ due to the unsuccessful computation
  $\esel{(\isel{\ok}{\ok})}{\prf{\coact{a}}{\nil}} \paral \prf{a}{\nil} \traSS{\acttau} \nil\paral\nil$. \exqed
\end{example}



\section{Monitors and Monitored Computations}
\label{sec:monitorability}
\begin{figure}[t]
	
	\textbf{Syntax}
	\begin{align*}
	\vV,\vVV \in \Verd\  \bnfdef\ &  \yes && (\text{acceptance}) && \bnfsepp  \no && (\text{rejection})\\
	\bnfsepp & \stp && (\text{inconclusive}) \\
	\mAct \in \textsc{Patterns} \ \bnfdef\ & \act && (\text{action}) && \bnfsepp \notAct{\act} && (\text{complement})\\
	\mV,\mVV\in\Mon\ \bnfdef\ & \vV && \text{(verdict)} && \bnfsepp  \prf{\mAct}{\mV} && \text{(interaction)} \\
	\bnfsepp & \esel{\mV}{\mVV} && \text{(choice)} && \bnfsepp  \mtimes{\mV}{\mVV} && \text{(conjunction)} 
	\end{align*}
	
	\textbf{Dynamics}
	\begin{mathpar}
		\inference[\rtit{mVer}]{}{\vV \traSS{\act} \vV}\and
		\inference[\rtit{mAct}]{}{\prf{\act}{\mV}  \traSS{\act} \mV} \and
		\inference[\rtit{mNAct}]{\actt \neq \act}{\prf{\notAct{\act}}{\mV}  \traSS{\actt} \mV}
		\\
		\inference[\rtit{mSelL}]{\mV \traSS{\act} \mV'}{\esel{\mV}{\mVV}  \traSS{\act} \mV'} \and
		\inference[\rtit{mSelR}]{\mVV \traSS{\act} \mVV'}{\esel{\mV}{\mVV}  \traSS{\act} \mVV'} 
		\and
		\inference[\rtit{mConj}]{\mV \traSS{\act} \mV' & \mVV \traSS{\act} \mVV'}{\mtimes{\mV}{\mVV}  \traSS{\act} \mtimes{\mV'}{\mVV'}} 
	\end{mathpar}
	
	\textbf{Instrumentation}
	\begin{mathpar}
		\inference[\rtit{iMon}]{p \traSS{\act} p' & \mV \traSS{\act} \mV'}{\sys{\mV}{p} \traSS{\act} \sys{\mV'}{p'}} 
		\and
		\inference[\rtit{iTer}]{p \traSS{\act} p' & \mV \traSSN{\act}  
                }{\sys{\mV}{p} \traSS{\act} \sys{\stp}{p'}}
		\and
		\inference[\rtit{iAsy}]{p \traSS{\acttau} p'}{\sys{\mV}{p} \traSS{\acttau} \sys{\mV}{p'}}
	\end{mathpar}
	
	\caption{Monitors and Instrumentation}
	\label{fig:monit-instr}
	\end{figure}

Figure~\ref{fig:monit-instr} describes  the monitoring framework used to analyse servers purporting to adhere to some advertised contract.
It defines the syntax of these monitors, which follow the general structure used in earlier works \cite{FrancalanzaAI:2015:uHML,BLC:2011:RVL:2000799.2000800:RV-LTL} whereby monitors may reach any one of the three verdicts \Verd, namely acceptance, rejection, or the inconclusive verdict. In addition to the basic prefixing patterns used in  \cite{FrancalanzaAI:2015:uHML,Fra16}, we here also use action complementation, $\notAct{\act}$, to denote any action \emph{apart from} $\act$.  As in \cite{FrancalanzaAI:2015:uHML,Fra16}, a monitor is allowed to branch, \esel{\mV}{\mVV}, depending on the actions observed at runtime.  We also find it convenient to express a merge monitor operator that facilitates the composition of monitor specifications, \mtimes{\mV}{\mVV}.  

The semantics of a monitor is given in terms of the LTS defined by the rules in Figure~\ref{fig:monit-instr}. This is best viewed as the evolution of a monitor in response to a (finite) execution trace $t \in \Act^\ast$, consisting of a sequence of actions $\act_1,\ldots,\act_n$.  Verdicts are irrevocable when reached, and do not change upon viewing further actions in the trace (rule \rtit{mVer}). Prefixing releases the guarded monitor when the expected pattern is encountered (rules \rtit{mAct} and \rtit{mNAct}). The rules  \rtit{mSelL} and \rtit{mSelR} describe left and right monitor branching as expected, whereas rule \rtit{mConj} describes the synchronous evolution of merged monitors.

A \emph{monitored server contract} consists of a server $p$ that is instrumented with a monitor \mV, denoted as \sys{\mV}{p}.  The behaviour of monitored contracts is defined as an LTS through the rules stated in Figure~\ref{fig:monit-instr}, and relies on the \resp LTSs of the monitor and the server.  Rule \rtit{iMon} states that if a server can transition with action $\act$ and the  monitor can follow this by transitioning with the same action, then in an instrumented server  they transition in lockstep. However, if the monitor cannot follow such a transition the instrumentation forces it to terminate with an inconclusive verdict, \stp, while the process is allowed to proceed unaffected; see rule \rtit{iTer}.  Finally, rule \rtit{iAsy} allows a contract to evolve independently from the monitor when performing silent $\tau$ moves (which are unobservable to the monitor).  We refer to a sequence of transitions from a monitored contract as a \emph{monitored computation} and use the standard notation $\sys{m}{p} \wtraSS{\,t\,} \sys{m'}{p'}$ that abstracts over $\tau$-moves in trace $t$. 

A few comments are in order.  First, we highlight the fact that in the operational semantics for monitored systems of \figref{fig:monit-instr},  the monitor does \emph{not} have access to the internal state of the server generating the trace, and its observations are limited to the execution that the server chooses to exhibit at runtime.  This is meant to model the RV scenarios mentioned in \secref{sec:introduction}, where the source of the executing  system cannot be analysed: from the point of view of the runtime monitoring and verification, the server description is merely used to generate traces.  Second, we note that, in a monitored server setup, 
any visible behaviour is instigated by the server, relegating the instrumented monitor to a \emph{passive} role that merely follows the server actions. Stated otherwise, the server  \emph{drives} the behaviour in a monitored system  and dictates  the execution path that the monitor can analyse at runtime.

In what follows, we explain how monitors work through a series of examples. The exposition focuses on monitors that produce rejection verdicts, but the discussion can be extended to acceptance verdicts in a straightforward manner.

\begin{example} \label{ex:monit-contract}
   The monitor $\esel{\prf{\notAct{\coact{a}}}{\no}}{\prf{\coact{a}}{\prf{\notAct{b}}{\no}}}$ checks for violations from  contracts that are expected to adhere to (\ie be supercontracts of) the contract  $\prf{\coact{a}}{\prf{b}{\nil}}$. In fact, the monitor reaches a rejection verdict whenever  a contract either emits an action that is not $\coact{a}$ at runtime, \prf{\notAct{\coact{a}}}{\no}, or else  emits an action that is not $\coact{b}$ following action $a$, \prf{\coact{a}}{\prf{\notAct{b}}{\no}}.  Consider the server $\prf{\coact{a}}{\prf{c}{\nil}}$; when instrumented with our monitor we can observe the following monitored computation whereby the monitor reaches a rejection verdict, \no.
   \begin{align*}
     & \sys{\bigl(\esel{\prf{\notAct{\coact{a}}}{\no}}{\prf{\coact{a}}{\prf{\notAct{b}}{\no}}}\bigr)}{(\prf{\coact{a}}{\prf{c}{\nil}})} \quad\traSS{\coact{a}}\quad \sys{\prf{\notAct{b}}{\no}}{\prf{c}{\nil}} \quad\traSS{c}\quad \sys{\no}{\nil}
   \end{align*}
   By contrast, when the server  $\prf{\coact{a}}{\prf{b}{\nil}}$ is instrumented with the monitor, no rejection verdict is reached; in particular, the final transition below is derived using rule \rtit{iTer} because $\prf{\notAct{b}}{\no} \traSSN{b}$.
   \begin{align*}
     & \sys{\bigl(\esel{\prf{\notAct{\coact{a}}}{\no}}{\prf{\coact{a}}{\prf{\notAct{b}}{\no}}}\bigr)}{(\prf{\coact{a}}{\prf{b}{\nil}})} \quad\traSS{\coact{a}}\quad \sys{\prf{\notAct{b}}{\no}}{\prf{b}{\nil}} \quad\traSS{b}\quad \sys{\stp}{\nil}
   \end{align*}
We emphasise the fact that monitor termination through rule \rtit{iTer} is crucial to avoid unwanted detections.  Consider a variant of the earlier  monitor,  $\prf{\coact{a}}{\prf{b}{\no}}$, which now reports violations whenever it observes the trace consisting of the action $\coact{a}$ followed by the action $b$. When composed with the system $\prf{\coact{a}}{\prf{c}{\prf{b}{\nil}}}$ we observe the following monitored computation.
\begin{align*}
   \sys{\prf{\coact{a}}{\prf{b}{\no}}}{\prf{\coact{a}}{\prf{c}{\prf{b}{\nil}}}} & \traSS{\coact{a}} \sys{\prf{b}{\no}}{\prf{c}{\prf{b}{\nil}}} \\
    & \quad\traSS{c}  \sys{\stp}{\prf{b}{\nil}}  \tag{**}\label{eq:iter} \\
    & \qquad\traSS{b} \sys{\stp}{\nil}
\end{align*}
   At transition (\ref{eq:iter}), the server can perform an action, $c$, that the monitor is not able to follow (\ie it is \emph{not} specified how the monitor should behave at that point should it observe action $c$).  Accordingly, the semantics instructs the monitor to terminate (prematurely) with an inconclusive verdict.   There are two instrumentation alternatives that could have been adopted, both of which are arguably wrong from a monitoring perspective.  The first option would have been to prohibit the server from exhibiting action $c$, which goes against the tenet that the monitor should adopt a passive role and not interfere with the execution of the program it monitors.  The second option is arguably even worse: we could have let the server transition and left the monitor in its present state, \ie    
   \begin{math}
     \sys{\prf{b}{\no}}{\prf{c}{\prf{b}{\nil}}} \traSS{c}  \sys{\prf{b}{\no}}{\prf{b}{\nil}}  
   \end{math}, 
   but then this would have led to an \emph{unspecified/erroneous} detection at the next transition 
   \begin{math}
     \sys{\prf{b}{\no}}{\prf{b}{\nil}}   \traSS{b}  \sys{\no}{\nil}  
   \end{math}.
    \exqed
\end{example}

\begin{example} \label{ex:monit-contract2}
  The server $\isel{\prf{\coact{a}}{\prf{b}{\nil}}}{\prf{c}{\prf{b}{\nil}}}$ is \emph{not} a supercontract of $\prf{\coact{a}}{\prf{b}{\nil}}$ according to \defref{def:server-pre}.  Crucially, however, in an RV setting, monitor detection depends on the runtime behaviour exhibited by the server.  This contrasts with other forms of verification which may be allowed to explore \emph{all} the execution paths of a server under scrutiny.\footnote{In the general case, a pre-deployment verification technique may also analyse \emph{infinite} paths.}
  \begin{align*}
    & \sys{\bigl(\esel{\prf{\notAct{\coact{a}}}{\no}}{\prf{\coact{a}}{\prf{\notAct{b}}{\no}}}\bigr)}{\bigl(\isel{\prf{\coact{a}}{\prf{b}{\nil}}}{\prf{c}{\prf{b}{\nil}}}  \bigr)}  \quad\traSS{\tau}\quad \sys{\bigl(\esel{\prf{\notAct{\coact{a}}}{\no}}{\prf{\coact{a}}{\prf{\notAct{b}}{\no}}}\bigr)}{(\prf{\coact{a}}{\prf{b}{\nil}})} \quad\traSS{\coact{a}}\quad \sys{\prf{\notAct{b}}{\no}}{\prf{b}{\nil}} \quad\traSS{b}\quad \sys{\stp}{\nil}\\
    & \sys{\bigl(\esel{\prf{\notAct{\coact{a}}}{\no}}{\prf{\coact{a}}{\prf{\notAct{b}}{\no}}}\bigr)}{\bigl(\isel{\prf{\coact{a}}{\prf{b}{\nil}}}{\prf{c}{\prf{b}{\nil}}}  \bigr)}  \quad\traSS{\tau}\quad \sys{\bigl(\esel{\prf{\notAct{\coact{a}}}{\no}}{\prf{\coact{a}}{\prf{\notAct{b}}{\no}}}\bigr)}{(\prf{c}{\prf{b}{\nil}})} \quad\traSS{c}\quad \sys{\no}{\prf{b}{\nil}} \quad\traSS{b}\quad \sys{\no}{\nil}
  \end{align*}
  In the first monitored computation above, the server exhibits the behaviour described by the trace $\wtraSS{\coact{a}b}$, which prohibits the monitor from detecting any violations. However, the same server exhibits a different trace $\wtraSS{cb}$ in the second monitored computation which permits monitor detection.  The rejection verdict is in fact reached after the first visible transition on action $c$, and then preserved throughout the remainder of the computation. \exqed

   
\end{example}

\begin{example} \label{ex:monit-contract-plus}
  We can monitor for violations of the contract $\esel{\prf{\coact{a}}{\prf{b}{\nil}}}{\prf{c}{\nil}}$ by composing two submonitors that monitor for the constituents.  Specifically, since the monitor $\esel{\prf{\notAct{c}}{\no}}{\prf{c}{\stp}}$ checks for violations of contract \prf{c}{\nil} and, the minimally extended monitor $\esel{\prf{\notAct{\coact{a}}}{\no}}{\prf{\coact{a}}{(\esel{\prf{\notAct{b}}{\no}}{\prf{b}{\stp}})}}$ checks for violations of \prf{\coact{a}}{\prf{b}{\nil}} as discussed in \exref{ex:monit-contract}, we can construct the composite monitor $\mtimes{(\esel{\prf{\notAct{\coact{a}}}{\no}}{\prf{\coact{a}}{(\esel{\prf{\notAct{b}}{\no}}{\prf{b}{\stp}})}})}{(\esel{\prf{\notAct{c}}{\no}}{\prf{c}{\stp}})}$ to monitor for violations of $\esel{\prf{\coact{a}}{\prf{b}{\nil}}}{\prf{c}{\nil}}$.  
  \begin{align*}
     \sys{\bigl(\mtimes{(\esel{\prf{\notAct{\coact{a}}}{\no}}{\prf{\coact{a}}{(\esel{\prf{\notAct{b}}{\no}}{\prf{b}{\stp}})}})}{(\esel{\prf{\notAct{c}}{\no}}{\prf{c}{\stp}})}\bigr)}{\esel{\prf{\coact{a}}{\prf{b}{\nil}}}{\prf{c}{\nil}}}  &\quad\traSS{\coact{a}}\quad  \sys{\mtimes{\esel{\prf{\notAct{b}}{\no}}{\prf{b}{\stp}}}{\no}}{\prf{b}{\nil}} \\
    & \qquad\traSS{b}\quad \sys{\mtimes{\stp}{\no}}{\nil}\\
    \sys{\bigl(\mtimes{(\esel{\prf{\notAct{\coact{a}}}{\no}}{\prf{\coact{a}}{(\esel{\prf{\notAct{b}}{\no}}{\prf{b}{\stp}})}})}{(\esel{\prf{\notAct{c}}{\no}}{\prf{c}{\stp}})}\bigr)}{\esel{\prf{\coact{a}}{\prf{b}{\nil}}}{\prf{c}{\nil}}}  & \quad\traSS{c}\quad \sys{\mtimes{\no}{\stp}}{\nil}
  \end{align*}
  When the composite monitor is instrumented with the contract it is expected to monitor for, we note that it does not reach a rejection along \emph{every} (parallel) submonitor.  
  \begin{align*}
    & \sys{\bigl(\mtimes{(\esel{\prf{\notAct{\coact{a}}}{\no}}{\prf{\coact{a}}{(\esel{\prf{\notAct{b}}{\no}}{\prf{b}{\stp}})}})}{(\esel{\prf{\notAct{c}}{\no}}{\prf{c}{\stp}})}\bigr)}{\prf{b}{\nil}} \quad\traSS{b}\quad  \sys{\mtimes{\no}{\no}}{\nil}
  \end{align*}
  By contrast, the violating contract above generates a rejection along every submonitor. \exqed 
\end{example}

 \exref{ex:monit-contract-plus} clearly suggests a definition of monitor rejection.

\begin{definition}[Rejection] \label{def:rej} A monitor $\mV$ is in a rejection state, denoted as \rej{\mV}, whenever it is of the form $\mtimes{\no}{\mtimes{\ldots}{\no}}$.  We overload this predicate to denote a server $p$ being rejected by a monitor $\mV$, defined formally as  
	\begin{align*}
		\rej{p,\mV}  &\deftxt \exists t,p' \cdot\; \sys{\mV}{p}\wtraS{t}\sys{\mV'}{p'} \text{ and } \rej{\mV'}
	\end{align*}
\end{definition}

\begin{example} \label{ex:monit-rej}  The monitor $\esel{\prf{\notAct{c}}{\no}}{\prf{c}{\stp}}$  rejects server $\prf{b}{\nil}$, $\rej{\prf{b}{\nil},(\esel{\prf{\notAct{c}}{\no}}{\prf{c}{\stp}})}$ as well as server $\esel{\prf{c}{\nil}}{\prf{b}{\nil}}$, $\rej{(\esel{\prf{c}{\nil}}{\prf{b}{\nil}}),(\esel{\prf{\notAct{c}}{\no}}{\prf{c}{\stp}})}$ because both may exhibit an execution trace that leads the monitor to a rejection state.  By contrast, $\esel{\prf{\notAct{c}}{\no}}{\prf{c}{\stp}}$  does \emph{not} reject server $\prf{c}{\nil}$. Recalling monitor $\mV=\bigl(\mtimes{(\esel{\prf{\notAct{\coact{a}}}{\no}}{\prf{\coact{a}}{(\esel{\prf{\notAct{b}}{\no}}{\prf{b}{\stp}})}})}{(\esel{\prf{\notAct{c}}{\no}}{\prf{c}{\stp}})}\bigr)$ from \exref{ex:monit-contract-plus}, we can also state that it rejects server $\prf{b}{\nil}$, \rej{\prf{b}{\nil},\mV}. \exqed
\end{example}



\section{Preliminary results towards Monitorability}
\label{sec:preliminary-results}
Monitorability may be broadly described as the relationship between the properties of a logic specifying program behaviour and the detection capabilities of a monitoring setup instrumented over such programs.   It is therefore parametric with respect to the logic and monitoring setup considered.    In what follows, we sketch out preliminary investigations that focus on  the monitor rejections defined in Section~\ref{sec:monitorability}, and  attempt to relate them to violations of the server preorder defined in Section~\ref{sec:language}.

We have already defined enough machinery to be able to state formally two important properties.  Definition~\ref{def:smon} states that a monitor $\mV$ \emph{soundly monitors} for a server contract $p$ if and only if, whenever it rejects a server $q$, it is indeed the case that $q$ is not a supercontract of $p$.  In a sense, the dual of this is  Definition~\ref{def:cmon}, which states that a monitor $\mV$ \emph{completely monitors} for a server contract $p$ if and only if  every $q$ that is not a supercontract of $p$ is rejected by \mV.

\begin{definition}[Rejection Sound] \label{def:smon}
	\begin{math}
	\smon{p,\mV} \deftxt\; \forall q \cdot \;\rej{q,\mV}\;\text{implies} \;p \not\srvPre q
	\end{math}. \exqed
\end{definition}

\begin{definition}[Rejection Complete] \label{def:cmon}
	\begin{math} 
	\cmon{p,\mV} \deftxt\; \forall q \cdot \; p \not\srvPre q \;\text{implies} \;\rej{q,\mV} 
	\end{math}. \exqed
\end{definition}

\noindent
We can also extend these monitorability definitions to a specification language of contracts (\ie a set of contracts). 

\begin{definition}[Language Rejection Monitorability] \label{def:mon-lang}  A set of contracts $\mathcal{C}$ is: 
  \begin{itemize}
  \item  sound rejection-monitorable iff
    $\forall p \in \mathcal{C} \cdot \exists m\in\Mon \cdot \smon{p,\mV}$
  \item  complete rejection-monitorable iff
    $\forall p \in \mathcal{C} \cdot \exists m\in\Mon \cdot \cmon{p,\mV}$ 
  \item rejection-monitorable iff $\forall p \in \mathcal{C} \cdot \exists m\in\Mon \cdot \smon{p,\mV} \text{ and }\cmon{p,\mV}$ \exqed
  \end{itemize}
\end{definition}

We can readily argue in a formal manner that the contract language \Srv of Figure~\ref{fig:syntax-semantics} \emph{cannot} be   rejection-monitorable. Consider as an example $\esel{\prf{\coact{a}}{\nil}}{\prf{b}{\nil}} \in \Srv$. If this language is  rejection-monitorable, then there must exist a monitor \mV such that $\smon{\esel{\prf{\coact{a}}{\nil}}{\prf{b}{\nil}},\mV}$ and $\cmon{\esel{\prf{\coact{a}}{\nil}}{\prf{b}{\nil}},\mV}$. We argue towards a contradiction. From Section~\ref{sec:language} we know that $\esel{\prf{\coact{a}}{\nil}}{\prf{b}{\nil}} \not\srvPre \prf{\coact{a}}{\nil}$, and thus, by $\cmon{\esel{\prf{\coact{a}}{\nil}}{\prf{b}{\nil}},\mV}$,  it must be the case that $\rej{\prf{\coact{a}}{\nil}, \mV}$. Now this rejection predicate holds if either \mV reaches a rejection state immediately or else reaches rejection after observing action $a$.   In either case, this monitor would also reject the contract $\esel{\prf{\coact{a}}{\nil}}{\prf{b}{\nil}}$ as well, which would make the monitor necessarily unsound, \ie  $\neg\smon{\esel{\prf{\coact{a}}{\nil}}{\prf{b}{\nil}},\mV}$, since, by the reflexivity property of the preorder, we have $\esel{\prf{\coact{a}}{\nil}}{\prf{b}{\nil}} \srvPre \esel{\prf{\coact{a}}{\nil}}{\prf{b}{\nil}}$.

We deem sound rejection to be the minimum correctness requirement to be expected from the contract monitors we consider.  Note, however, that the contract language \Srv of Figure~\ref{fig:syntax-semantics} is trivially sound rejection-monitorable via the monitor $\stp$; this monitor never reaches a rejection state and thus trivially satisfying $\rej{p,\stp}$ for any $p\in\Srv$.  However, we argue that this monitor, \stp,  is not very useful.

We attempt to go one step further and define an automated monitor synthesis function that returns a monitor for \emph{every server} in the contract language \Srv.  We argue, at least informally, that these synthesised monitors are, in some sense, useful because they perform a degree of violation detections.  Importantly, however, we show that these synthesised monitors are rejection sound, according to Definition~\ref{def:smon}. 

\begin{definition}[Monitor Synthesis]\label{def:mon-synthesis} The function $\monSyn{-}: \Srv \rightarrow \Mon$ synthesises a monitor from a server contract description, and is defined inductively on the structure of this contract as follows: 
	\begin{align*}
	\qquad\qquad\qquad&&\monSyn{\nil} & \deftxt \stp
	\quad &\monSyn{\prf{\alpha}{p}} &\deftxt \esel{\prf{\notAct{\alpha}}{\no}}{\prf{\alpha}{\monSyn{p}}}\\
	&&\monSyn{\esel{p}{q}} &\deftxt \mtimes{\monSyn{p}}{\monSyn{q}} \qquad\qquad
	&\monSyn{\isel{p}{q}} &\deftxt \mtimes{\monSyn{p}}{\monSyn{q}}  & & \qquad\qquad\qquad\qquad\qquad\;\;\exqed
	\end{align*}
\end{definition}

A few comments on Definition~\ref{def:mon-synthesis} are in order.  First, note that a number of the monitors considered earlier in  Section~\ref{sec:monitorability} are in fact instances of this translation.  For instance, we have
\begin{equation*}
  \monSyn{\prf{\coact{a}}{\prf{b}{\nil}}} = \esel{\prf{\notAct{\coact{a}}}{\no}}{\prf{\coact{a}}{(\esel{\prf{\notAct{b}}{\no}}{\prf{b}{\stp}})}} \qquad \text{and}\qquad \monSyn{\esel{\prf{\coact{a}}{\prf{b}{\nil}}}{\prf{c}{\nil}}} = \mtimes{(\esel{\prf{\notAct{\coact{a}}}{\no}}{\prf{\coact{a}}{(\esel{\prf{\notAct{b}}{\no}}{\prf{b}{\stp}})}})}{(\esel{\prf{\notAct{c}}{\no}}{\prf{c}{\stp}})}
\end{equation*}
from \exref{ex:monit-contract-plus}.  Secondly, note that the monitor synthesis does not attempt to perform any detection violation for the contract $\nil$.  Since $\nil$ is  in some sense
a bottom element in the preorder, no supercontract of $\nil$ is allowed to perform any visible action.  Thus, in cases where all the actions permissible in \Srv are known up front as a  \emph{finite} set $\sset{\act_1,\ldots,\act_n}$, we can improve the precision of our synthesis through the alternative definition 
\begin{math}
  \monSyn{\nil} \deftxt \esel{\prf{\notAct{\act_1}}{\no}}{\esel{\ldots}{\prf{\notAct{\act_n}}{\no}}} 
\end{math}
for the case where $p = \nil$.  Third, note that the synthesis for both internal and external choice constructs coincide which, in a sense, is due to the  inherent  discriminating limits of RV.  Consider, by way of example, the monitor syntheses below:
\begin{equation*}
  \monSyn{\esel{\prf{\coact{a}}{\nil}}{\prf{b}{\nil}} } \; = \; \mtimes{((\esel{\prf{\notAct{\coact{a}}}{\no}}{\prf{\coact{a}}{\stp}}))}{((\esel{\prf{\notAct{b}}{\no}}{\prf{b}{\stp}}))} \;=\; \monSyn{\isel{\prf{\coact{a}}{\nil}}{\prf{b}{\nil}}}
\end{equation*}
The server $\prf{c}{\nil}$ is rejected by the monitor $\mtimes{((\esel{\prf{\notAct{\coact{a}}}{\no}}{\prf{\coact{a}}{\stp}}))}{((\esel{\prf{\notAct{b}}{\no}}{\prf{b}{\stp}}))}$ and accordingly it is \emph{neither} a supercontract of  $\esel{\prf{\coact{a}}{\nil}}{\prf{b}{\nil}}$ \emph{nor} of  $\isel{\prf{\coact{a}}{\nil}}{\prf{b}{\nil}}$.  However, the server  $\prf{\coact{a}}{\nil}$ is \emph{not} rejected by the monitor $\mtimes{((\esel{\prf{\notAct{\coact{a}}}{\no}}{\prf{\coact{a}}{\nil}}))}{((\esel{\prf{\notAct{b}}{\no}}{\prf{b}{\nil}}))}$; whereas it is correct to do so in the case of monitoring for the internal choice contract $\isel{\prf{\coact{a}}{\nil}}{\prf{b}{\nil}}$ because $\isel{\prf{\coact{a}}{\nil}}{\prf{b}{\nil}} \srvPre \prf{\coact{a}}{\nil}$, it leads to lack of precision in the case of the external choice $\esel{\prf{\coact{a}}{\nil}}{\prf{b}{\nil}}$ since $\esel{\prf{\coact{a}}{\nil}}{\prf{b}{\nil}} \not\srvPre \prf{\coact{a}}{\nil}$.  In spite of these limitations, we are able to show that our proposed monitor synthesis is sound.

\begin{theorem}[Synthesis Soundness]\label{thm:synth-sound} For every server specification $p \in \Srv$, every server implementation $q \in \Srv$, and the monitor synthesis function \monSyn{-} of Definition~\ref{def:mon-synthesis}:
\begin{equation*}\text{Whenever }\; \rej{q, \monSyn{p}} \;\text{ then it is necessarily the case that }\; p \not\srvPre q
\end{equation*}  
\end{theorem}
\begin{proof}
  By structural induction on the server specification $p$.
\end{proof}


\section{Conclusion}
\label{sec:conclusion}

We have presented preliminary investigations relating to the monitorability of contracts, high-level descriptions for web services.  We developed a monitoring framework that complements the operational semantics of server contracts.  We then focused on the rejection expressivity of the monitors within this framework and related it to  cases where it is unsafe to replace one server (contract) with another.  Within our simple framework, we were already able to identify limits with respect to monitor detection powers, and were able to diagnose problems with a proposed automated monitor synthesis procedure.  We were also able to formally prove that, in spite of its limit, the monitor synthesis considered is, in some sense, correct (Theorem~\ref{thm:synth-sound}).  

\paragraph*{Related and Future Work}
\label{sec:related-work}
The language of contracts for web services has been discussed in several other works prior to ours, such as \cite{Bernardi:2012:MST:2245276.2232097:Contracts,DBLP:journals/mscs/BravettiZ09a,Padovani:2009:CDAWS:contracts,Castagna:2009:TCW:1538917.1538920:Contracts}; although conceptually simple, it has been shown to be expressive enough to capture the dynamicity of interactions specified by more elaborate contract descriptions.
The server preorder considered in this paper captures the essence of the must preorder, studied in \cite{DBLP:journals/corr/BernardiH15} and the compliance preorder, studied in \cite{Padovani:2009:CDAWS:contracts,Castagna:2009:TCW:1538917.1538920:Contracts};  in our simplistic case of finite servers and clients, the two preorders coincide (modulo minor technical details regarding client satisfaction and computation success).  Our notion of monitorability is inspired by that presented in \cite{FrancalanzaAI:2015:uHML}, which relates process satisfaction of a branching-time logic, $p\models\phi$, with detections of monitors synthesised from formulas in this logic, $\sys{\monSyn{\phi}}{p}$.  The instrumentation relation considered in this paper is in fact an adaptation to the one used in \cite{FrancalanzaAI:2015:uHML}. 

For future work, we aim to achieve a more comprehensive study of monitorability than the preliminary one presented in Section~\ref{sec:preliminary-results}. In particular, we plan to consider monitor acceptances as a verdict in addition to rejections, establish stronger results with respect to rejections and consider extended contract descriptions similar to  \cite{DBLP:journals/corr/BernardiH15,Castagna:2009:TCW:1538917.1538920:Contracts} that include recursion and the potential for infinite computation.  This will lead to different notions of server refinements such as those resulting from compliance and fair testing preorders  \cite{DBLP:journals/mscs/BravettiZ09a,Padovani:2009:CDAWS:contracts}.  It will be interesting to study whether any of the aforementioned server preorder variants are more monitorable than the others.




\paragraph*{Acknowledgements:}
\label{sec:acknowledgements}
This research was partly supported by the project ``TheoFoMon: Theoretical Foundations for Monitorability''  of the Icelandic Research Fund.

\bibliographystyle{eptcs} 
\bibliography{references}



\end{document}